\documentclass{llncs}%[a4wide]{llncs}
\usepackage{epsfig,amsmath,amssymb,verbatim}%,a4wide}
\usepackage[ruled,vlined]{algorithm2e}
\usepackage{float}
\usepackage{multirow}
\usepackage{mdwlist}
\usepackage{subfig}

\newtheorem{observation}{Observation}

\usepackage{xspace}

\newcommand{\Reals}{{\mathbb{R}}}            % real numbers
\newcommand{\eps}{\varepsilon}               % epsilon
\newcommand{\NP}{{\normalfont NP}}
\newcommand{\forloop}[5][1]%

\newcommand{\pstart}{s}
\newcommand{\pend}{e}
\newcommand{\opdist}{\overline{p}}
\newcommand{\cdist}{c}
\newcommand{\decvar}{\lambda}

\def\O{{\cal O}}

\def\W{{\cal W}}

\newcommand{\doap}{{\sc doap}}
\newcommand{\bcmd}{{\sc doap}\xspace}

\begin{document}

\title{Fast Algorithms for Diameter-Optimally Augmenting Paths and
Trees\thanks{A preliminary version of this paper appeared in the
Proceedings of the 42nd International Colloquium on Automata, Languages,
and Programming (ICALP), Part I, Lecture Notes in Computer Science,
Vol.\ 9134, Springer-Verlag, Berlin, 2015, pp.\ 678--688.
M.S.\ was supported by NSERC. J.G. was supported by the ARC’s Discovery Projects funding scheme (DP150101134).} }
\author{Ulrike Gro{\ss}e\inst{1} \and Joachim Gudmundsson\inst{2}
\and Christian Knauer\inst{1} \and Michiel Smid\inst{3} \and Fabian Stehn\inst{1}}

\institute{Institut f\"ur Angewandte Informatik, Universit\"at Bayreuth,  Bayreuth, Germany
%\email{ulrike.grosse|christian.knauer|fabian.stehn@uni-bayreuth.de}
\and
School of Information Technology, University of Sydney, Sydney, Australia
%\email{joachim.gudmundsson@sydney.edu.au}
\and
School of Computer Science, Carleton University, Ottawa, Canada
%\email{michiel@scs.carleton.ca}
}

\maketitle

\begin{abstract}
We consider the problem of augmenting an $n$-vertex graph embedded
in a metric space, by inserting one additional edge in order to
minimize the diameter of the resulting graph. We present exact
algorithms for the cases when (i) the input graph is a path, running in
$O(n \log^3 n)$ time, and (ii) the input graph is a tree, running in
$O(n^2 \log n)$ time. We also present an algorithm that computes a
$(1+\eps)$-approximation in $O(n + 1/\eps^3)$ time, for paths in
$\Reals^{d}$, where $d$ is a constant.
\end{abstract}

% JG - rewrote the introduction. Please improve!

\section{Introduction}
Let $G=(V,E)$ be a graph in which each edge has a positive weight.
The weight (or length) of a path is the sum of the weights of the edges
on this path. For any two vertices $x$ and $y$ in $V$, we denote by
$\delta_G(x,y)$ their shortest-path distance, i.e., the minimum weight
of any path in $G$ between $x$ and $y$. The diameter of $G$ is defined
as $\max \{ \delta_G(x,y) : x,y \in V \}$.

Assume that we are also given weights for the non-edges of the graph $G$.
In the \emph{Diameter-Optimal $k$-Augmentation Problem}, \doap$(k)$, we have to
compute a set $F$ of $k$ edges in $(V \times V) \setminus E$ for which
the diameter of the graph $(V,E\cup F)$ is minimum.

%\begin{problem}\label{problem}
%Given an undirected graph $G=(V,E)$, a weight function
%$w:[V]^2 \rightarrow \Reals$, and a positive integer $k$; compute
%a set $F$ of $k$ non-edges so that the diameter of the graph
%$G_F=(V,E\cup F)$ with weight function $w$ is minimized.
%\end{problem}

In this paper, we assume that the given graph is a path or a tree
on $n$ vertices that is embedded in a metric space, and the weight of
any edge and non-edge is equal to the distance between its vertices. We
consider the case when $k=1$; thus, we want to compute one non-edge
which, when added to the graph, results in an augmented graph of
minimum diameter. Surprisingly, no non-trivial results were known even
for the restricted cases of paths and trees.

Throughout the rest of the paper, we assume that $(V,|\cdot|)$ is a
metric space, consisting of a set $V$ of $n$ elements (called points
or vertices). The distance between any two points $x$ and $y$ is
denoted by $|xy|$. We assume that an oracle is available that returns
the distance between any pair of points in $O(1)$ time. Our contributions
are as follows:
\begin{enumerate}
\item If $G$ is a path, we solve problem \doap$(1)$ in $O(n \log^3 n)$
      time.
\item If $G$ is a path and the metric space is $\Reals^d$, where $d$ is
      a constant, we compute a $(1+\eps)$-approximation for \doap$(1)$ in
      $O(n + 1/\eps^3)$ time.
\item If $G$ is a tree, we solve problem \doap$(1)$ in $O(n^2 \log n)$
      time.
\end{enumerate}

\subsection{Related Work}
%THERE ARE SOME UNDEFINED REFERENCES

The Diameter-Optimal $k$-Augmentation Problem for edge-weighted graphs,
and many of its variants, have been shown to be
NP-hard~\cite{sbl-diced-97}, or even
$W[2]$-hard~\cite{fggm-agmd-13,ghn-pcgda-13}. Because of this, several
special classes of graphs have been considered.
Chung and Gary~\cite{cgdbag-84} and Alon et al.~\cite{agr-ddbdg-99}
considered paths and cycles with unit edge weights and gave upper and
lower bounds on the diameter that can be achieved.
Ishii~\cite{i-aogmd-13} gave a constant factor approximation algorithm
(approximating both $k$ and the diameter) for the case when the input
graph is outerplanar. Erd\H{o}s et al.~\cite{egr-hddtf-98} investigated
upper and lower bounds for the case when the augmented graph must be
triangle-free.

\paragraph{The general problem:}
The Diameter-Optimal Augmentation Problem can be seen as a bicriteria
optimization problem: In addition to the weight, each edge and non-edge
has a cost associated with it. Then the two optimization criteria are
(1) the total cost of the edges added to the graph and (2) the diameter
of the augmented graph. We say that an algorithm is an
$(\alpha,\beta)$-approximation algorithm for the \bcmd problem, with
$\alpha, \beta \geq 1$, if it computes a set $F$ of non-edges of total
cost at most $\alpha\cdot B$ such that the diameter of $G'=(V,E\cup F)$
is at most $\beta\cdot D^B_\mathrm{opt}$, where $D^B_\mathrm{opt}$ is the diameter of
an optimal solution that augments the graph with edges of total cost at
most $B$.

For the restricted version when all costs and all weights are
identical~\cite{bgp-ianar-12,cv-atmcd-02,dk-dnbpd-99,ks-bdmcg-07,lms-mcbdb-92},
Bil\`{o} et al.~\cite{bgp-ianar-12} showed that, unless P=\NP, there
does not exist a $(c \log n, \delta<1+1/D^B_\mathrm{opt})$-approximation
algorithm for \bcmd if $D^B_\mathrm{opt}\geq 2$. For the case in which
$D^B_\mathrm{opt}\geq 6$, they proved that, again unless P=\NP, there does not
exist a $(c \log n, \delta<\frac {5}{3}-\frac {7-(D^B_\mathrm{opt}+1) \bmod 3} {3D^B_\mathrm{opt}})$-approximation algorithm.

Li et al.~\cite{lms-mcbdb-92} showed a $(1,4+ 2/D^B_\mathrm{opt})$-approximation
algorithm. The analysis of the algorithm was later improved by Bil\`{o}
et al.~\cite{bgp-ianar-12}, who showed that it gives a
$(1,2+2/D^B_\mathrm{opt})$-approximation. In the same paper they also gave an
 $(O(\log n), 1)$-approximation algorithm.

For general costs and weights,
Dodis and Khanna~\cite{dk-dnbpd-99} gave an
$O(n \log D^B_\mathrm{opt}, 1)$-approximation algorithm.
%, assuming that all weights are polynomially bounded.
Their result is based on a
multi-commodity flow formulation of the problem.
Frati et al.~\cite{fggm-agmd-13} recently considered the \bcmd problem
with arbitrary integer costs and weights. Their main result is a
$(1,4)$-approximation algorithm with running time
$O((3^B B^3 + n + \log (Bn)) B n^2)$.

\paragraph{Geometric graphs:}
In the geometric setting, when the input is a geometric graph embedded
in the Euclidean plane, there are only a few results on graph
augmentation in general. Rutter and Wolff~\cite{rw-acpgg-12} proved
that the $k$-connectivity and $k$-edge-connectivity augmentation problems
are NP-hard on plane geometric graphs, for $k = 2, 3, 4$, and $5$; the
problem is infeasible for $k\geq 6$ because every planar graph has a
vertex of degree at most 5. Currently, there are no known approximation
algorithms for this problem. Farshi et al.~\cite{fgg-isfgn-05} gave
approximation algorithms for the problem of adding one edge to a
geometric graph while minimizing the dilation. There were several
follow-up papers~\cite{lw-cbwsg-08,w-cdeag-10}, but there is still
no non-trivial result known for the case when $k > 1$.

In the \emph{continuous} version of the diameter-optimal augmentation
problem, the input graph $G$ is embedded in the plane and the edges
to be added to $G$ can have their endpoints anywhere \emph{on} $G$,
i.e., the endpoints can be in the interior of edges of $G$.
Moreover, the diameter is considered as the maximum of the shortest-path
distances over all points \emph{on} $G$.
Yang~\cite{ecats} considered the continuous version of the problem of
adding one edge to a path so as to minimize the continuous diameter.
He presented sufficient and necessary conditions for an augmenting
edge to be optimal. He also presented an approximation algorithm,
having an \emph{additive} error of $\epsilon$, that runs in
$O(( n + |P|/\epsilon )^2 n )$ time, where $|P|$ denotes the length
of the input path $P$ and $\epsilon$ is at most half of the length of a
shortest edge in $P$. De Carufel et al.~\cite{mcdwapc} improved the
running time to $O(n)$ and also considered the continuous version
of the problem for cycles that are embedded in the plane. They showed
that adding one edge to any cycle does not decrease the continuous
diameter. On the other hand, two edges can always be added that decrease
the continuous diameter. De Carufel et al.\ gave a full characterization
of the optimal two edges. If the input cycle is convex, they find the
optimal pairs of edges in $O(n)$ time.

\section{Augmenting a Path with One Edge} \label{sec:diam-path-exact}
We are given a path $P = (p_1, \ldots,p_n)$ on $n$ vertices in a
metric space and assume that it is stored in an array $P[1,\dots,n]$. To
simplify notation, we associate a vertex with its index, that is
$p_k=P[k]$ is also referred to as $k$ for $1\leq k\leq n$. This allows
us to extend the total order of the indices to the vertex set of $P$.
We denote the start vertex of $P$ by $\pstart$ and the end vertex of $P$
by $\pend$.
%, i.e., $\pstart=P[1]$, $\pend=P[n]$.

For $1 \le k < l \le n$, we denote the subpath $(p_k, \dots, p_l)$ of
$P$ by $P[k,l]$, the cycle we get by adding the edge $\overline{p_kp_l}$
to $P[k,l]$ by $C[k,l]$, and the (unicyclic) graph we get by adding the
edge $\overline{p_kp_l}$ as a {\em shortcut} to $P$ by
$\overline{P}[k,l]$; the length of $X\in \{P,P[k,l],C[k,l]\}$ is denoted
by~$|X|$. We will consider the functions
$\opdist_{k,l} := \delta_{\overline{P}[k,l]}$ and
$\cdist_{k,l} := \delta_{C[k,l]}$, where $\delta_G$ is the
length of the shortest path between two vertices in $G$.
For $1 \le k < l \le n$, we let
\[ M(k,l) := \max_{1 \le x < y \le n} \opdist_{k,l}(x,y)\]
denote the {\em diameter} of the graph $\overline{P}[k,l]$.

Our goal is to compute a shortcut $\overline{p_kp_l}$ for $P$ that
minimizes the diameter of the resulting unicyclic graph, i.e., we want
to compute
\[ m(P) := \min_{1 \le k < l \le n} M(k,l).
\]
We will prove the following result:

\begin{theorem}  \label{thm:main}
Given a path $P$ on $n$ vertices in a metric space, we can compute
$m(P)$, and a shortcut realizing that diameter, in $O(n \log^3 n)$ time.
\end{theorem}
The algorithm consists of two parts. We first describe a sequential
algorithm for the {\em decision problem}. Given $P$ and
a threshold parameter $\decvar>0$, decide if
$m(P) \le \decvar$ (see Lemma~\ref{lem:seqdecalg}~a) below). In a second
step, we argue that the sequential algorithm can be implemented in a
parallel fashion (see Lemma~\ref{lem:seqdecalg}~b) below), thus enabling
us to use the parametric search paradigm of Megiddo.

\begin{lemma}
\label{lem:seqdecalg}
Given a path $P$ on $n$ vertices in a metric space
and a real parameter $\decvar>0$, we
can decide in
\begin{enumerate*}
	\item[\textbf{a)}] $O(n \log n)$ time, or in
	\item[\textbf{b)}] $O(\log n)$ parallel time using $n$ processors
\end{enumerate*}
 whether $m(P)\le \decvar$; the algorithms
also produce a feasible shortcut if it exists.
\end{lemma}

To prove this lemma, observe that
\[ m(P) \le \decvar \text{ if and only if }
        \bigvee_{1 \le k < l \le n} M(k,l) \le \decvar .
\]
The algorithm checks, for each $1 \le k < n$, whether there is
some $k < l \le n$ such that $M(k,l) \le \decvar$. If one such index
$k$ is found, we know that $m(P) \le \decvar$; otherwise
$m(P) > \decvar$. Clearly this approach also produces a feasible shortcut
if it exists.

We decompose the function $M(k,l)$ into four monotone parts. This will
facilitate our search for a feasible shortcut and enable us to do
(essentially) binary search: For $1 \le k < l \le n$, we let
\begin{align*}
S(k,l) &:= \max_{k \le x \le l} \opdist_{k,l}(\pstart,x),\quad
&E(k,l) &:= \max_{k \le x \le l} \opdist_{k,l}(x,\pend),\\
U(k,l) &:= \opdist_{k,l}(\pstart,\pend),\quad
&O(k,l) &:= \max_{k \le x < y \le l} \cdist_{k,l}(x,y).
\end{align*}
Then we have $M(k,l) = \max \{S(k,l), E(k,l), U(k,l), O(k,l)\}$. The triangle inequality
implies that
\begin{align*}
S(k,l) &\le S(k,l+1),\quad
&E(k,l) &\ge E(k,l+1),\\
U(k,l) &\ge U(k,l+1),\quad
&O(k,l) &\le O(k,l+1).
\end{align*}
%$$m(P) = \min \{ \decvar \mid m(P) \le \decvar \}.$$
\begin{figure}[tb]
        \subfloat[\label{fig:path_diameter}]{
                \includegraphics[width=.46\textwidth]{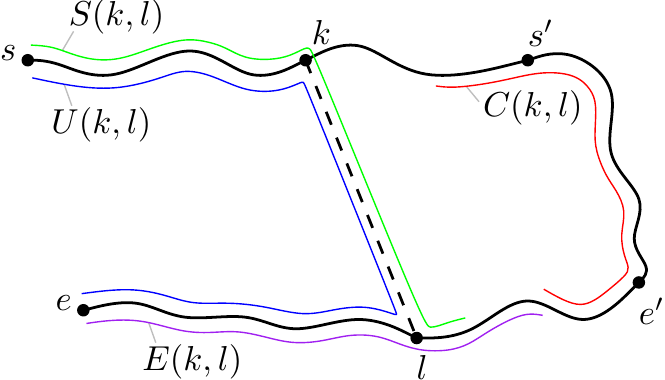}}\hspace*{5ex}
        \subfloat[\label{fig:path_circledistance}]{
                \includegraphics[width=.46\textwidth]{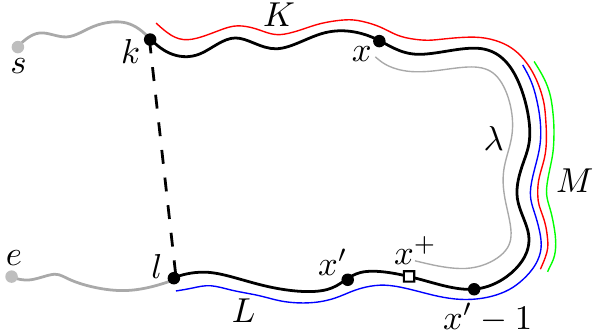}}
\caption{\textbf{(a)}~Illustration of the four distances that define the diameter of a shortcut $\overline{p_kp_l}$: $U(k,l)$ is the length of the shortest path connecting $s$ and $e$; $O(k,l)$ is the length of the longest shortest path between any two points in $C[k,l]$; $S(k,l)$ ($E(k,l)$) is the length of the longest shortest path from $s$ ($e$) to any vertex in $C(k,l)$.
\textbf{(b)}~Illustration of the computation of $O(k,l)$.}
\label{fig:path}
\end{figure}
The function $U$ is easy to evaluate once we have the array
$D[1,\dots,n]$ of the prefix-sums of the edge lengths:
$D[i] := \sum_{1 \le j < i} |p_jp_{j+1}|$.
These sums can be computed in $O(n)$ time sequentially %in a sequential fashion
or in $O(\log n)$ time using $n$ processors.
If in addition to $D$, the vertices
% edited 19.2.'15
$s' = \max\{v\,|\,\delta_P(s,v) \leq \decvar\}$ and
$e'= \min\{v\,|\,\delta_P(v,e) \leq \decvar\}$ are computed for a fixed
$\decvar$ in $O(\log n)$ time (via binary search on $D$), the following
decision problems can be answered in constant time:
\begin{align*}
S(k,l) \le \decvar,\quad
E(k,l) \le \decvar,\quad
U(k,l) \le \decvar.
\end{align*}
We denote the maximum of these three functions by
$$N(k,l) = \max (S(k,l), E(k,l), U(k,l)).$$
Now clearly
$$M(k,l) = \max (N(k,l), O(k,l))$$
and, consequently
$$M(k,l) \le \decvar \text{ if and only if } N(k,l) \le \decvar \text{ and } O(k,l) \le \decvar.$$

\noindent
For fixed $1 \le k < n$, the algorithm will first check whether there is some $k < l \le n$ with $N(k,l) \le \decvar$. If no such $l$ exists, we can conclude
that $M(k,l) > \decvar$ for all $k < l \le n$.
The monotonicity of $S$, $E$, and $U$ implies that, for fixed $1 \le k < n$, the set
$$N_k := \{k < l \le n \mid N(k,l) \le \decvar\}$$
is an {\em interval}. This interval can be computed (using binary search in $P$ and in $D$ as described above) in $O(\log n)$ time. If $N_k = \emptyset$ we can conclude that for the $1 \le k < n$ under consideration and for all $k < l \le n$, we have that $M(k,l) > \decvar$.

If $N_k$ is non-empty, the monotonicity of $O$ implies that it is
sufficient to check for $l_k = \min N_k$ (i.e. the starting point of
the interval) whether $O(k,l_k) \le \decvar$:
$$\exists k < l \le n : O(k,l) \le \decvar\text{ if and only if } O(k,l_k) \le \decvar.$$
Note that in this case we know that $N(k,l_k) \le \decvar$.

\paragraph{Deciding the diameter of small cycles:}
We now describe how to decide for a given shortcut $1 \le k < l \le n$ if $O(k,l)\le \decvar$, given that {\em we already know that} $N(k,l)\le \decvar$. To this end, consider the following sets of vertices from $C[k,l]$:
$K :=\{k \le x \le l \mid \delta_P(k,x) \le \decvar \}$,
$L :=\{k \le x \le l \mid \delta_P(x,l) \le \decvar \}$,
$M := K \cap L$,
$K' := K \setminus L$,
$L' := L \setminus K$.

These sets are intervals and can be computed in $O(\log n)$ time by binary search.
Since $N(k,l)\le \decvar$, we can conclude the following:
\begin{itemize}
%\item $|C[k,l]|\le 3\decvar$
\item the set of vertices of $C[k,l]$ is $K \cup L$
\item $\cdist_{k,l}(x,y) \le \decvar$ for all $x,y \in K$
\item $\cdist_{k,l}(x,y) \le \decvar$ for all $x,y \in L$
\item $\cdist_{k,l}(x,y) \le \decvar$ for all $x \in M$, $y \in C[k,l]$
\end{itemize}
\noindent
Consequently, if $\cdist_{k,l}(x,y) > \decvar$ for $x,y \in C[k,l]$,
we can conclude that $x \in K'$ and $y \in L'$. In order to establish
that $O(k,l) \le \decvar$, it therefore suffices to verify that
$$\bigwedge_{x \in K', y \in L'} \cdist_{k,l}(x,y)\le\decvar.$$
%$$\max \{\cdist_{k,l}(x,y)\mid x \in K', y \in L'\} \le \decvar.$$
Note that on $P$ any vertex $x$ of $K'$ is at least $\decvar$ away from the vertex $l$, i.e., $\delta_P(x,l) > \decvar $.
Let $x^+$ be point on (a vertex or an edge of) $P$ that is closer (along $P$) by a distance of $\decvar$ to $l$ than to $x$, i.e., $x^+$ is the unique point on $P$ such that
$$\delta_P(x^+,l) < \delta_P(x,l) \text{ and } \delta_P(x,x^+) = \decvar.$$
The next (in the direction of $l$) {\em vertex} of $P$ will be denoted by $x'$,
i.e., $x < x'\le l$ is the unique vertex of $P$ such that
$$\delta_P(x,x'-1) \le \decvar \text{ and } \delta_P(x,x') > \decvar.$$
Since $x$ is a vertex of $K'$, $x'$ is a vertex of $L'$.
For the following discussion we denote the distance achieved in $C[k,l]$ by using the shortcut by $\cdist^+_{k,l}$ and the distance achieved by travelling along $P$ only by $\cdist^-_{k,l}$, i.e.,
$$\cdist^-_{k,l}(x,y) := \delta_P(x,y)\text{ and }\cdist^+_{k,l}(x,y) := \delta_P(x,k)+|\overline{p_kp_l}|+\delta_P(l,y).$$
Clearly
\[\cdist_{k,l}(x,y) = \min (\cdist^+_{k,l}(x,y), \cdist^-_{k,l}(x,y)), \text{ and }
|C[k,l]|= \cdist^+_{k,l}(x,y)+\cdist^-_{k,l}(x,y).\]
For every vertex $y < x'$ on $L'$ we have that $\cdist_{k,l}(x,y) \le \cdist^-_{k,l}(x,y) \le \decvar$, so if there is some vertex $x' \neq y \in L'$ such that $\cdist_{k,l}(x,y) > \decvar$, we know that $x' < y \le l$; in that case we have that
$\cdist^+_{k,l}(x,y) \le \cdist^+_{k,l}(x,x')$.
Since we assume that $\cdist_{k,l}(x,y) > \decvar$, we also know that $\cdist^+_{k,l}(x,y) > \decvar$
and we can conclude that $\cdist^+_{k,l}(x,x') > \decvar$, and consequently that
$\cdist_{k,l}(x,x') > \decvar$, i.e., for all $x \in K'$ we have that
$$\bigwedge_{y \in L'} \cdist_{k,l}(x,y)\le \decvar   \text{ if and only if } \cdist_{k,l}(x,x') \le  \decvar.$$

The distance between (the point) $x^+$ and (the vertex) $x'$ on $P$ is called the {\em defect} of $x$ and is denoted by $\Delta(x)$, i.e., $\Delta(x) = \delta_P(x^+,x')$.

\begin{lemma}
We have
$$\cdist_{k,l}(x,x') \le \decvar \text{ if and only if }
|C[k,l]| \le \Delta(x) + 2\decvar .$$
\end{lemma}
\begin{proof}
Observe that
\begin{align*}
|C[k,l]| & = \delta_P(x,k)+|\overline{p_kp_l}|+\delta_P(l,x')+\delta_P(x',x^+)+\delta_P(x^+,x)\\
&= \delta_P(x,k)+|\overline{p_kp_l}|+\delta_P(l,x')+\Delta(x)+\decvar\\
& = \cdist^+_{k,l}(x,x')+\Delta(x)+\decvar.
%& > \Delta(x)+2\decvar
\end{align*}
Since $\cdist^-_{k,l}(x,x') > \decvar$, we have that $\cdist_{k,l}(x,x') \le \decvar$ if and only if $\cdist^+_{k,l}(x,x') \le \decvar$; the claim follows.
\qed
\end{proof}

To summarize the above discussion, we have the following chain of equivalences (here $\Delta_{k,l} := |C[k,l]| - 2\decvar$):

\begin{equation*}
O(k,l) \le \decvar
\Leftrightarrow \bigwedge_{x \in K'} \cdist_{k,l}(x,x')\le\decvar\Leftrightarrow \bigwedge_{x \in K'}\Delta_{k,l} \le \Delta(x)\Leftrightarrow \min_{x \in K'} \Delta(x) \ge\Delta_{k,l}.
\end{equation*}
Since $K'$ is an interval, the last condition can be tested easily after some preprocessing: To this end we compute a $1d$-range tree on $D$ and associate with each vertex in the tree the minimum $\Delta$-value of the corresponding canonical subset. For every {\em vertex} $x$ of $P$ that is at least $\decvar$ away from the end vertex of $P$ we can compute $\Delta(x)$ in $O(\log n)$ time by binary search in $D$. With these values the range tree can be built in $O(n)$ time. A query for an interval $K'$ then gives us $\mu := \min_{x \in K'} \Delta(x)$ in $O(\log n)$ time and we can check  the above condition in $O(1)$ time.

We describe the algorithm in pseudocode; see Algorithm~\ref{alg:decisionalgorithm}.

\begin{algorithm}
%\DontPrintSemicolon
%\LinesNotNumbered
\SetKwFunction{DecisionAlgorithm}{\textsc{DecisionAlgorithm}}
\SetKwFunction{CheckOForShortcut}{\textsc{CheckOForShortcut}}
\SetKwFunction{ComputePrefixSums}{\textsc{ComputePrefixSums}}
\SetKwFunction{ComputeSegmentTree}{\textsc{ComputeRangeTree}}
\SetKwFunction{ComputeFeasibleIntervalForN}{\textsc{ComputeFeasibleIntervalForN}}
\caption{Algorithm for deciding if $m(P) \le \decvar$}
\label{alg:decisionalgorithm}

\BlankLine
\DecisionAlgorithm{$P,\decvar$} \tcp*[r]{Decide if $m(P) \le \decvar$}
\nl\Begin{
    \textbf{global} $D \leftarrow $\ \ComputePrefixSums{$P$}\;
    \textbf{global} $s' \leftarrow \max\{v\,|\,\delta_P(s,v)\leq \decvar\}$\;
    \textbf{global} $e' \leftarrow \min\{v\,|\,\delta_P(v,e)\leq \decvar\}$\;
    \textbf{global} $T \leftarrow $\ \ComputeSegmentTree{$P,\decvar$}\;
    \For{$1 \le k <n$}
    {
        $N_k \leftarrow $\ \ComputeFeasibleIntervalForN{$k,\decvar$}\;
        \If{$N_k \neq \emptyset$ {\bf and} \CheckOForShortcut{$k,\min(N_k),\decvar$}}
        {
            \KwRet{\sc True}
        }
    }
    \KwRet{\sc False}
}
\BlankLine
\CheckOForShortcut{$k,l,\decvar$} \tcp*[r]{Decide if $O(k,l) \le \decvar$}
\nl\Begin{

    $K' \leftarrow \{k \le x \le l \mid \delta_P(k,x) \le \decvar \wedge \delta_P(x,l) > \decvar\}$\tcp*[r]{Compute the interval by binary search}
    $\mu \leftarrow \min_{x \in K'} \Delta(x)$ \tcp*[r]{Query the range tree $T$}
  \KwRet{ $(\mu \ge |C[k,l]| - 2\decvar)$}
}
\end{algorithm}
The correctness of the algorithm follows from the previous discussion.
{\sc ComputePrefixSums} runs in $O(n)$ time,
{\sc ComputeRangeTree} runs in $O(n \log n)$ time,
{\sc ComputeFeasibleIntervalForN} runs in $O(\log n)$ time, a call to
{\sc CheckOForShortcut} requires $O( \log n)$ time.
The total runtime is therefore $O( n \log n)$.
It is easy to see that with $n$ processors, the steps {\sc ComputePrefixSums} and {\sc ComputeRangeTree} can be realized in $O(\log{n})$ parallel time and that with this number of processors, all calls to {\sc CheckOForShortcut} can be handled in parallel. Therefore, the entire algorithm can be parallelized and has a parallel runtime of $O(\log n)$, as stated in Lemma~\ref{lem:seqdecalg}~b).
This concludes the proof of Lemma~\ref{lem:seqdecalg}.

When we plug this result into the %(standard)
parametric search technique
of Megiddo, we get the algorithm for the optimization problem as claimed
in Theorem~\ref{thm:main}.

From the above discussion, we note that, since there are only four
possible distances to compute to determine the diameter of a path
augmented with one shortcut edge, the following corollary follows
immediately.

\begin{corollary} \label{cor:compute diameter}
Given a path $P$ on $n$ vertices in a metric space and a shortcut
$(u,v)$, the diameter of $P \cup (u,v)$ can be computed in $O(n)$ time.
\end{corollary}

\section{An Approximation Algorithm in Euclidean Space}
\label{sec:diam-path-apx}

In Section~\ref{sec:diam-path-exact}, we presented an
$O(n \log^3 n)$-time algorithm for the problem when the input graph is
a path in a metric space. Here we show a simple $(1+\eps)$-approximation
algorithm with running time $O(n +1/\eps^3)$ for the case when the
input graph is a path in $\Reals{^d}$, where $d$ is a constant.
The algorithm will use two ideas: clustering and the well-separated
pair decomposition (WSPD) as introduced by
Callahan and Kosaraju~\cite{ck-dmpsa-95}.

\begin{definition} [\cite{ck-dmpsa-95}]   \label{wellsep}
       Let $s>0$ be a real number, and let $A$ and $B$ be two finite
       sets of points in $\Reals^d$. We say that $A$ and $B$ are
       \emph{well-separated} with respect to $s$, if there are two
       disjoint $d$-dimensional balls $C_A$ and $C_B$, having the
       same radius, such that
       (i) $C_A$ contains $A$,
       (i) $C_B$ contains $B$, and
       (ii) the minimum distance between $C_A$ and $C_B$ is at least
       $s$ times the radius of $C_A$.
\end{definition}

The parameter $s$ will be referred to as the {\em separation constant}.
The next lemma follows easily from Definition~\ref{wellsep}.

\begin{lemma} [\cite{ck-dmpsa-95}]  \label{insamepair}
       Let $A$ and $B$ be two finite sets of points that are
       well-separated w.r.t.\ $s$, let $x$ and $p$ be points of $A$,
       and let $y$ and $q$ be points of $B$. Then
       (i) $|xy| \leq (1+4/s) \cdot |pq|$, and
       (ii) $|px| \leq (2/s) \cdot |pq|$.
\end{lemma}

\begin{definition}[\cite{ck-dmpsa-95}]  \label{defWSPD}
       Let $S$ be a set of $n$ points in $\Reals^d$, and let $s>0$
       be a real number. A {\em well-separated pair decomposition}
       (WSPD) for $S$ with respect to $s$ is a sequence of pairs
       of non-empty subsets of~$S$,
       $(A_1,B_1) ,  \ldots , (A_m, B_m)$,
       such that
       \begin{enumerate}
       \item $A_i \cap B_i = \emptyset$, for all $i=1, \ldots, m$,
       \item for any two distinct points $p$ and $q$ of
             $S$, there is exactly one pair $(A_i,B_i)$
             in the sequence, such that
             (i) $p \in A_i$ and $q \in B_i$,
             or (ii) $q \in A_i$ and $p \in B_i$,
       \item $A_i$ and $B_i$ are well-separated w.r.t.\ $s$,
             for $1\leq i \leq m$.
       \end{enumerate}
       The integer $m$ is called the {\em size} of the WSPD.
\end{definition}

Callahan and Kosaraju showed that a WSPD of size $m = \O(s^dn)$ can

be computed in $\O(s^dn+n \log n)$ time.
\subsubsection{Algorithm}
We are given a polygonal path $P$ on $n$ vertices in $\Reals^{d}$.
We assume without loss of generality that the total length of $P$ is
$1$. Partition $P$ into $m=1/\eps_1$ subpaths $P_1, \ldots , P_m$,
each of length $\eps_1$, for some constant $0< \eps_1<1$ to
be defined later. Note that a subpath may have one (or both) endpoint
in the interior of an edge.
For each subpath $P_i$, $1\leq i\leq m$, select an
arbitrary vertex $r_i$ along $P_i$ as a representative vertex, if it
exists. The set of representative vertices is denoted $R_{P}$; note that
the size of this set is at most $m=1/\eps_1$.
Let $P(R)$ be the path consisting of the vertices of $R_P$, in the order
in which they appear along the path $P$. We give each edge $(u,v)$
of $P(R)$ a weight equal to $\delta_{P}(u,v)$. %(If $u$ or $v$ is in
the interior of an edge of $P$, then $\delta_{P}(u,v)$ is defined in
the natural way.)

Imagine that we ``straighten'' the path $P(R)$, so that it is
contained on a line. In this way, the vertices of this path form a
point set in $\Reals^{1}$; we compute a well-separated pair
decomposition $\W$ for the one-dimensional set $R_{P}$, with separation
constant $1/\eps_2$, with $0< \eps_2<1/4$ to be defined later.
Then, we go through all pairs $\{A,B\}$ in $\W$ and compute the
diameter of $P(R) \cup \{(rep(A),rep(B)\}$, where $rep(A)$ and $rep(B)$
are representative points of $A$ and $B$, respectively, which are
arbitrarily chosen from their sets. Note that
the number of pairs in $\W$ is $O(1/\eps_1\eps_2)$.
Finally the algorithm outputs the best shortcut.

\subsubsection{Analysis}
We first discuss the running time and then turn our attention to
the approximation factor of the algorithm.

The clustering takes $O(n)$ time, and constructing the WSPD of $R_{P}$
takes
$O(\frac {1}{\eps_1\eps_2}+\frac {1}{\eps_1} \log \frac {1}{\eps_1})$
time. For each of the $O(1/\eps_1\eps_2)$ well-separated pairs in $\W$,
computing the diameter takes, by Corollary~\ref{cor:compute diameter},
time linear in the size of the uni-cyclic graph, that is,
$O(\frac {1}{\eps_1^2\eps_2})$ time in total.

\begin{lemma} \label{cor:apx-running time}
  The running time of the algorithm is $O(n + \frac {1}{\eps_1^2\eps_2})$.
\end{lemma}

%In the analysis we will need the fact that adding a shortcut to a path can only decrease the diameter with at most a factor 2.
%\begin{lemma} \label{lem:Diameterbound} \cite{egl-apgmd-14}
% Let $P$ be a polygonal path and let $G=P\cup((u,v)$ be an optimal augmentation of $P$ with diameter $d_{opt}$. It holds that $diameter(P)\leq %3\cdot diameter(P)$.
%\end{lemma}

Before we consider the approximation bound, we need to define some
notation. Consider any vertex $p$ in $P$. Let $r(p)$ denote the
representative vertex of the subpath of $P$ containing $p$. For any
two vertices $p$ and $q$ in $P$, let $\{A,B\}$ be the well-separated
pair such that $r(p) \in A$ and $r(q)\in B$. The representative
points of $A$ and $B$ will be denoted $w(p)$ and $w(q)$, respectively.

\begin{lemma}
For any shortcut $e=(p,q)$ and for any two vertices $x,y \in P$,
we have
   $$(1-4\eps_2) \cdot \delta_{G}(x,y) - 6\eps_1 \leq \delta_{H}(w(x),w(y)) \leq (\frac {1}{1-4\eps_2}) \cdot \delta_{G}(x,y) + 6\eps_1,$$
   where $G=P \cup \{(p,q)\}$ and $H=P(R)\cup \{(w(p),w(q))\}$.
\end{lemma}
\begin{proof}
We only prove the second inequality, because the proof of the first
inequality is almost identical.

Consider two arbitrary vertices $x,y$ in $P$, and consider a shortest
path in $G$ between $x$ and $y$. We have two cases:\\
\noindent {\bf Case 1:} If $\delta_G(x,y)=\delta_{P}(x,y)$,
then $\delta_H(r(x),r(y)) \leq \delta_{P}(x,y) +2\eps_1$. \\
\noindent {\bf Case 2:} If $\delta_G(x,y) < \delta_{P}(x,y)$,
then the shortest path in $G$ between $x$ and $y$ must traverse
$(p,q)$. Assume that the path is
$x\rightsquigarrow p \rightarrow q \rightsquigarrow t$, thus
$\delta_G(x,y)=\delta_{P}(x,p)+|pq|+\delta_{P}(q,y)$. Consider the
following three observations:\\

  \noindent (1) $|pq| \geq |r(p)r(q)|-2\eps_1$ and $|w(p)w(q)|\leq (1+4\eps_2) \cdot |r(p)r(q)|$. Consequently, $|w(p)w(q)|\leq (1+4\eps_2) \cdot (|pq| + 2\eps_1)$. \\

  \noindent (2) We have
         \begin{eqnarray*}
          \delta_{P}(x,p) & \geq & \delta_{P}(w(x),w(p))-\delta_{P}(w(x),x) - \delta_{P}(w(p),p) \\
            & \geq & \delta_{P}(w(x),w(p)) - (\eps_1 + \delta_{P}(w(x),r(x))) - (\eps_1 + \delta_{P}(w(y),r(y)))\\
            & \geq & \delta_{P}(w(x),w(p)) - (\eps_1 + 2\eps_2 \delta_{P}(w(x),w(p))) - (\eps_1 + 2\eps_2 \delta_{P}(w(x),w(p)))\\
            &  =   & (1 - 4\eps_2) \cdot \delta_{P}(w(x),w(p)) - 2\eps_1\\
            &  \geq   & (1 - 4\eps_2) \cdot \delta_{H}(w(x),w(p)) - 2\eps_1 .
        \end{eqnarray*}
        That is, $\delta_{H}(w(x),w(p)) \leq \frac {1}{1-4\eps_2} \cdot \delta_{P}(x,p) + 2\eps_1$. \\

  \noindent (3) We have, $\delta_{H}(w(y),w(q)) \leq \frac {1}{1-4\eps_2} \cdot \delta_{P}(y,q) + 2\eps_1$, following the same arguments as in~(2).\\

  \noindent Putting together the three observations we get:
  \begin{eqnarray*}
     \delta_{H}(w(x),w(y)) & \leq & \delta_{H}(w(x),w(p))+|w(p)w(q)|+\delta_{H}(w(q),w(y))\\
        & \leq & (\frac {1}{1-4\eps_2}) \cdot \delta_{P}(x,p) + 2\eps_1 +
                       (1+4\eps_2) \cdot (|pq| + 2\eps_1) \\
        &&+ (\frac {1}{1-4\eps_2}) \cdot \delta_{P}(y,q) + 2\eps_1 \\
        & < & (\frac {1}{1-4\eps_2}) \cdot \delta_{G}(x,y) + 6\eps_1,
  \end{eqnarray*}
where the last inequality follows from the fact that $0<\eps_2<1/4$.
This concludes the proof of the lemma.
\qed
\end{proof}

By setting $\eps_1=\eps/60$ and $\eps_2=\eps/32$ and using the fact that the diameter of $H$ is at least $1/2$, we obtain the following theorem that summarizes this section.
\begin{theorem}
  Given a path $P$ with $n$ vertices in $\Reals^{d}$ and a real number
  $\eps>0$, we can compute a shortcut to $P$ in $O(n +1/\eps^3)$ time
  such that the resulting uni-cyclic graph has diameter at most
  $(1+\eps)\cdot d_\mathrm{opt}$, where $d_\mathrm{opt}$ is the diameter of an
  optimal solution.
\end{theorem}

\section{Augmenting a Tree with One Edge}
\label{sec:diam-tree-exact}

Next we consider the case when the input graph is a tree $T=(V,E)$,
where $V$ is a set of $n$ vertices in a metric space. The aim is to
compute an edge $f$ in $(V \times V ) \setminus E$ such that the diameter
of the resulting unicyclic graph $(V, E \cup {f})$ is minimized.

Let $P_T$ be the common intersection of all longest paths in $T$. Observe
that $P_T$ is a non-empty path in $T$. We denote the endvertices of
$P_T$ by $a$ and $b$. Let $F=T\setminus E(P_T)$ be the forest that
results from deleting the edges of $P_T$ from $T$. For any vertex $u$
of $T$,
\begin{enumerate}
\item let $\sigma(u)$ be the vertex on $P_T$ that is in the same tree of
      $F$ as $u$, and
\item let $\tau(u)$ be the tree of $F$ that contains $u$.
\end{enumerate}
Refer to Figure~\ref{fig:T} for an illustration.

\begin{figure}[t]
\centering
\includegraphics[width=8cm]{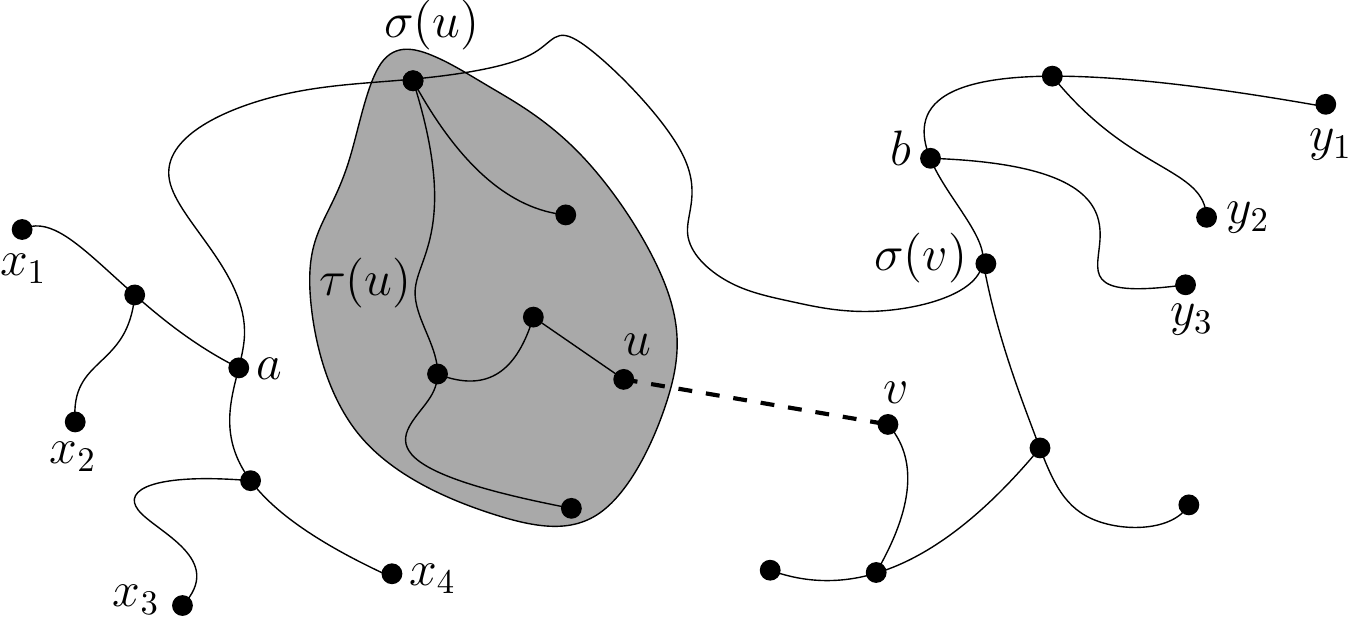}
\caption{Illustrating the input tree $T$ with $(u,v)$ as an optimal
        shortcut. The paths in $T$ between $x_i$ and $y_j$, for
        $1 \leq k \leq 4$ and $1 \leq j \leq 3$, represent all longest
        paths in $T$. These paths intersect in the path between $a$ and
        $b$.}
\label{fig:T}
\end{figure}

Consider any augmenting edge $(u,v)$. In the following lemma, we will
prove that the augmenting edge $(\sigma(u),\sigma(v))$ is at least as
good as $(u,v)$. That is, the diameter of
$T \cup \{(\sigma(u),\sigma(v))\}$ is at most the diameter of
$T \cup \{(u,v)\}$. In case $\sigma(u)=\sigma(v)$,
$T \cup \{(\sigma(u),\sigma(v))\}$ is equal to $T$, and the diameter
of $T \cup \{(u,v)\}$ is equal to the diameter of $T$.

\begin{lemma} \label{lem:vertices_on_path}
There exists an optimal augmenting edge $f$ for $T$ such that both
vertices of $f$ are vertices of $P_T$.
\end{lemma}
\begin{proof}
Consider an optimal augmenting edge $(u,v)$. We may assume without loss
of generality that $\sigma(u)$ is on the subpath of $P_T$ between $a$
and $\sigma(v)$. See Figure~\ref{fig:T}.

Let $T_{opt}=T\cup\{(u,v)\}$, let $D_{opt}$ be the diameter of $T_{opt}$,
and let $T'=T\cup\{(\sigma(u),\sigma(v))\}$. In order to prove the
lemma, it suffices to show that the diameter of $T'$ is at most
$D_{opt}$. If $D_{opt}$ is equal to the diameter of $T$, then this
obviously holds, because the diameter of $T'$ is at most the diameter
of $T$. Thus, from now on, we assume that $D_{opt}$ is less than
the diameter of $T$.

We claim that there exist endvertices $x$ and $y$ of some longest path
in $T$ such that
\begin{enumerate}
\item $a$ is on the path in $T$ between $x$ and $\sigma(u)$,
\item $b$ is on the path in $T$ between $y$ and $\sigma(v)$,
\item $x$ is not a vertex of $\tau(u) \setminus \{\sigma(u)\}$,
\item $y$ is not a vertex of $\tau(v) \setminus \{\sigma(v)\}$.
\end{enumerate}
To prove this, consider the leaves $x_1,x_2,\ldots,x_k$ and
$y_1,y_2,\ldots,y_{\ell}$ of $T$ such that
\begin{enumerate}
\item for each $i$ with $1 \leq i \leq k$, $a$ is on the path in $T$
      between $x_i$ and $b$,
\item for each $j$ with $1 \leq j \leq \ell$, $b$ is on the path in $T$
      between $y_j$ and $a$,
\item for each $i$ and $j$ with $1 \leq i \leq k$ and
      $1 \leq j \leq \ell$, the path in $T$ between $x_i$ and $y_j$ is a
      longest path in $T$, and each longest path in $T$ is between some
      $x_i$ and some $y_j$.
\end{enumerate}
Refer to Figure~\ref{fig:T}.
If $k=1$, then $x_1 = a$ and we take $x=x_1$. Assume that $k \geq 2$.
Consider the maximal subtree of $T$ that contains $a$ and all leaves
$x_1,x_2,\ldots,x_k$, and imagine this subtree to be rooted at $a$.
There is a child $a'$ of $a$ such that $u$ is not in the subtree
rooted at $a'$. We take $x$ to be any $x_i$ that is in the subtree
rooted at $a'$. By a symmetric argument, we can prove the existence
of the vertex $y$.

\vspace{0.5em}

Recall that we assume that the diameter of $T_{opt}$ (i.e., $D_{opt}$)
is less than the diameter of $T$. This implies that the shortest path in
$T_{opt}$ from $x$ to $y$ contains the shortcut $(u,v)$ and,
therefore,
\begin{equation} \label{eqHello}
   \delta_T(\sigma(u),u) + |uv| + \delta_T(v,\sigma(v)) <
   \delta_T(\sigma(u),\sigma(v)) .
\end{equation}
In particular, $\sigma(u) \neq \sigma(v)$.

Now let $s$ and $t$ be any pair of vertices. In the rest of the proof,
we will show that $\delta_{T'}(s,t)\leq D_{opt}$. Up to symmetry, there
are three main cases to consider with respect to the positions of $s$
and $t$:

\begin{enumerate}
 \item \textit{Both vertices are in trees of $F$ that contain the
shortcut vertices:} $s,t \in \tau(u) \cup \tau(v)$, see
Fig~\ref{fig:Case1}.
  \begin{enumerate}
    \item \textit{The vertices are in different trees of $F$:}
         $s\in \tau(u)$ and $t \in \tau(v)$.\newline
      Since
\[ \delta_{T'}(s,\sigma(u)) = \delta_T(s,\sigma(u))
      \leq \delta_T(x,\sigma(u)) = \delta_{T'}(x,\sigma(u))
\]
and
\[ \delta_{T'}(\sigma(v),t) = \delta_T(\sigma(v),t)
      \leq \delta_T(\sigma(v),y) = \delta_{T'}(\sigma(v),y),
\]
we have
\begin{eqnarray*}
 \delta_{T'}(s,t) & = & \delta_{T'}(s,\sigma(u)) +
                        | \sigma(u) \sigma(v) | +
                        \delta_{T'}(\sigma(v),t) \\
     & \leq & \delta_{T'}(x,\sigma(u)) + | \sigma(u) \sigma(v) | +
              \delta_{T'}(\sigma(v),y) \\
     & = & \delta_{T'}(x,y) \\
     & \leq & \delta_{T_{opt}}(x,y) \\
     & \leq & D_{opt} .
\end{eqnarray*}

    \item \textit{The vertices are in the same tree of $F$:}
           $s,t \in \tau(u)$.\newline
We will prove that in this case, the shortest paths between $s$ and
$t$ in both $T'$ and $T_{opt}$ do not contain the shortcut, i.e., both
these shortest paths are equal to the path in $\tau(u)$ (and, thus, in
$T$) between $s$ and $t$. This will imply that
\[ \delta_{T'}(s,t) = \delta_{T_{opt}}(s,t) \leq D_{opt} .
\]

Consider the shortest path $P'(s,t)$ between $s$ and $t$ in $T'$.
Observe that shortest paths do not contain repeated vertices.
If $P'(s,t)$ contains the shortcut $(\sigma(u),\sigma(v))$, then this
path visits the vertex $\sigma(u)$ twice. Thus, $P'(s,t)$ does not
contain $(\sigma(u),\sigma(v))$.

Consider the shortest path $P_{opt}(s,t)$ from $s$ to $t$ in
$T_{opt}$, and assume that this path contains $(u,v)$. We may assume
without loss of generality that, starting at $s$, this path traverses
$(u,v)$ from $u$ to $v$. (Otherwise, we interchange $s$ and $t$.)
Since $P_{opt}(s,t)$ does not contain repeated vertices, this path
contains the subpath in $T$ from $\sigma(v)$ to $\sigma(u)$. This
subpath must be the shortest path in $T_{opt}$ between $\sigma(v)$ and
$\sigma(u)$. However, as we have seen in (\ref{eqHello}), this is not
the case. Thus, we conclude that $P_{opt}(s,t)$ does not contain
$(u,v)$.
  \end{enumerate}

\begin{figure}[t]
\centering
\includegraphics[width=\textwidth]{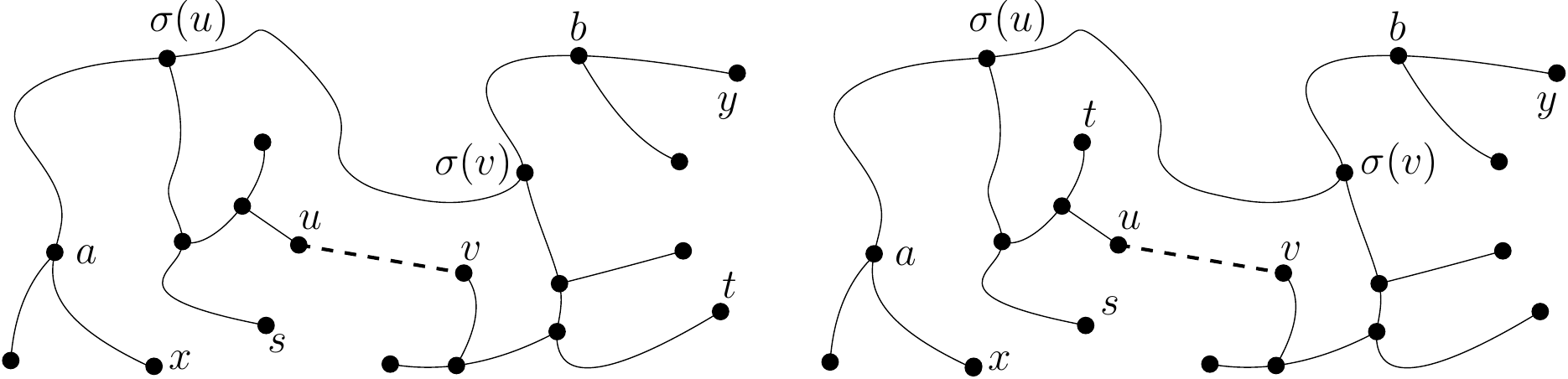}
\caption{Illustrating (left) case 1(a) and (right) case 1(b).}
\label{fig:Case1}
\end{figure}

  \item \textit{Neither vertices are in trees of $F$ that contain the
         shortcut vertices:} $s,t \notin \tau(u) \cup \tau(v)$.
    \newline
If the shortest path in $T_{opt}$ from $s$ to $t$ does not contain
$(u,v)$, then
\[ \delta_{T'}(s,t) \leq \delta_T(s,t) = \delta_{T_{opt}}(s,t)
            \leq D_{opt} .
\]
Assume that this shortest path contains $(u,v)$. We may assume without
loss of generality that this shortest path traverses the edge $(u,v)$
from $u$ to $v$. We have
\begin{eqnarray*}
 \delta_{T'}(s,t) & \leq &
   \delta_T(s,\sigma(u)) + | \sigma(u) \sigma(v) | +
      \delta_T(\sigma(v),t)  \\
 & \leq & \delta_T(s,\sigma(u)) +
          \delta_T(\sigma(u),u) + |uv| + \delta_T(v,\sigma(v))
          + \delta_T(\sigma(v),t)  \\
 & = & \delta_{T_{opt}}(s,t) \\
 & \leq & D_{opt} .
\end{eqnarray*}

  \item \textit{One vertex is in a tree of $F$ that contains a shortcut
    vertex, the other is not:} $s\in \tau(u)$ and
    $t\notin \tau(u) \cup \tau(v)$.
   \begin{enumerate}
    \item \textit{$t$ is a vertex in the maximal subtree of $T$ having
        $x$ and $\sigma(u)$ as leaves}, see Fig.~\ref{fig:Case3}(left).
      \newline
    As in Case 1(b), it can be shown that the shortest paths between
    $s$ and $t$ (as well as the shortest paths between $s$ and $x$)
    in both $T_{opt}$ and in $T'$ do not contain the shortcut. Thus,
    \[ \delta_{T'}(s,t) \leq \delta_{T'}(s,x) = \delta_{T}(s,x) =
          \delta_{T_{opt}}(s,x) \leq D_{opt}.
    \]
    \item \textit{$t$ is a vertex in the maximal subtree of $T$ having
          $\sigma(u)$ and $\sigma(v)$ as leaves}, see
          Fig.~\ref{fig:Case3}(right).
           \newline
We first observe that
\begin{eqnarray*}
   \delta_{T'}(s,t) & = & \delta_T(s,\sigma(u)) +
                          \delta_{T'}(\sigma(u),t) \\
    & \leq & \delta_T(x,\sigma(u)) + \delta_{T'}(\sigma(u),t) \\
    & = & \delta_{T'}(x,t) .
\end{eqnarray*}
If the shortest path in $T_{opt}$ from $x$ to $t$ does not contain
$(u,v)$, then
\[ \delta_{T'}(x,t) \leq \delta_T(x,t) = \delta_{T_{opt}}(x,t)
         \leq D_{opt} .
\]
Assume that the shortest path in $T_{opt}$ from $x$ to $t$ contains
$(u,v)$. Then
\[ \delta_{T_{opt}}(x,t) = \delta_T(x,\sigma(u)) +
                           \delta_T(\sigma(u),u) + |uv| +
                           \delta_T(v,\sigma(v)) +
                           \delta_T(\sigma(v),t) .
\]
Observe that
\[ \delta_{T'}(x,t) \leq \delta_T(x,\sigma(u)) +
                            | \sigma(u) \sigma(v) | +
                            \delta_T(\sigma(v),t) .
\]
The triangle inequality implies that
\[ \delta_{T'}(x,t) \leq \delta_{T_{opt}}(x,t) \leq D_{opt} .
\]
    \item $t$ is a vertex in the maximal subtree of $T$ having
          $\sigma(v)$ and $y$ as leaves.
          \newline
          In this case, we have
    \begin{eqnarray*}
        \delta_{T'}(s,t)
          & = & \delta_T(s,\sigma(u)) + | \sigma(u) \sigma(v) | +
                \delta_T(\sigma(v),t) \\
          & \leq & \delta_T(x,\sigma(u)) + | \sigma(u) \sigma(v) | +
                   \delta_T(\sigma(v),y) \\
          & \leq & \delta_T(x,\sigma(u)) +
                   \delta_{T_{opt}}(\sigma(u),\sigma(v)) +
                   \delta_T(\sigma(v),y) \\
          & = & \delta_{T_{opt}}(x,y) \\
          & \leq & D_{opt} .
    \end{eqnarray*}
   \end{enumerate}
\end{enumerate}
This concludes the proof of the lemma.
\hfill $\square$ \end{proof}

   \begin{figure}[t]
\centering
\includegraphics[width=\textwidth]{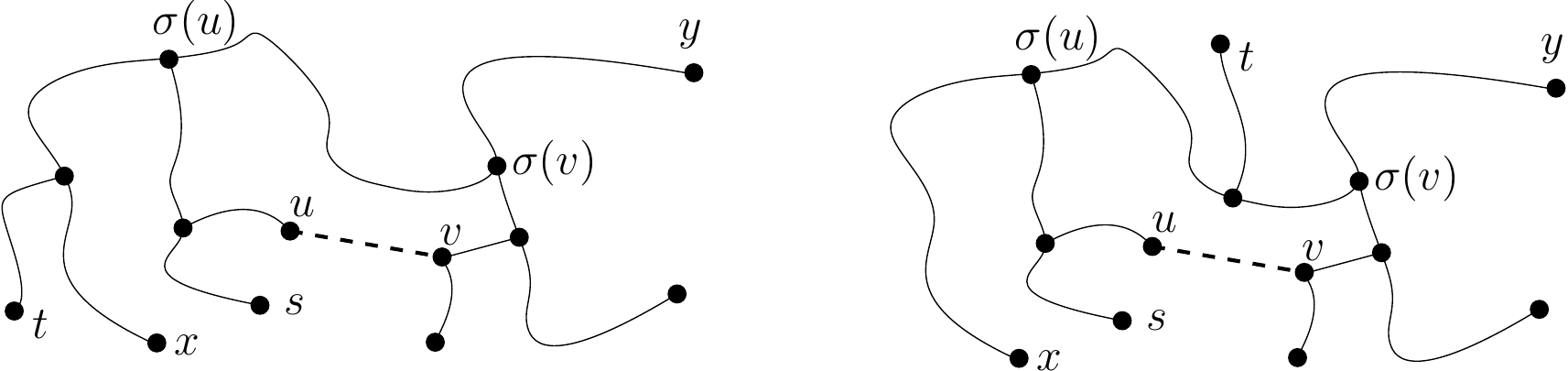}
\caption{Illustrating (left) case 3(a) and (right) case 3(b).}
\label{fig:Case3}
\end{figure}

\begin{comment}
\begin{corollary}\label{cor:longestPaths}
If the tree has several shortest paths, then there is an optimal
shortcut where both endpoints lie on the common intersection of all
shortest paths.
\end{corollary}
\end{comment}

As a consequence of Lemma~\ref{lem:vertices_on_path}, the diameter of
a tree cannot be improved by adding a single shortcut, if the
intersection of all longest paths is a vertex or a single edge.

\subsection{Augmenting a tree}
For a tree $T$ with $n$ vertices, let the intersection of all longest
paths in $T$ be the path $P_T$. In a preprocessing step, we convert
$T$ to a caterpillar tree $T_{cp}$ by replacing every tree $T'$ of
$T\setminus E(P_T)$ by a single edge of length $\delta_{T}(t,v)$,
where $v$ is the common vertex of $T'$ and $P_T$, and $t$ is the
furthest vertex in $T'$ to $v$, see Figure~\ref{fig:Caterpillar}. Note
that $T_{cp}$ has a unique longest path.

   \begin{figure}[ht]
\centering
\includegraphics[width=.7\textwidth]{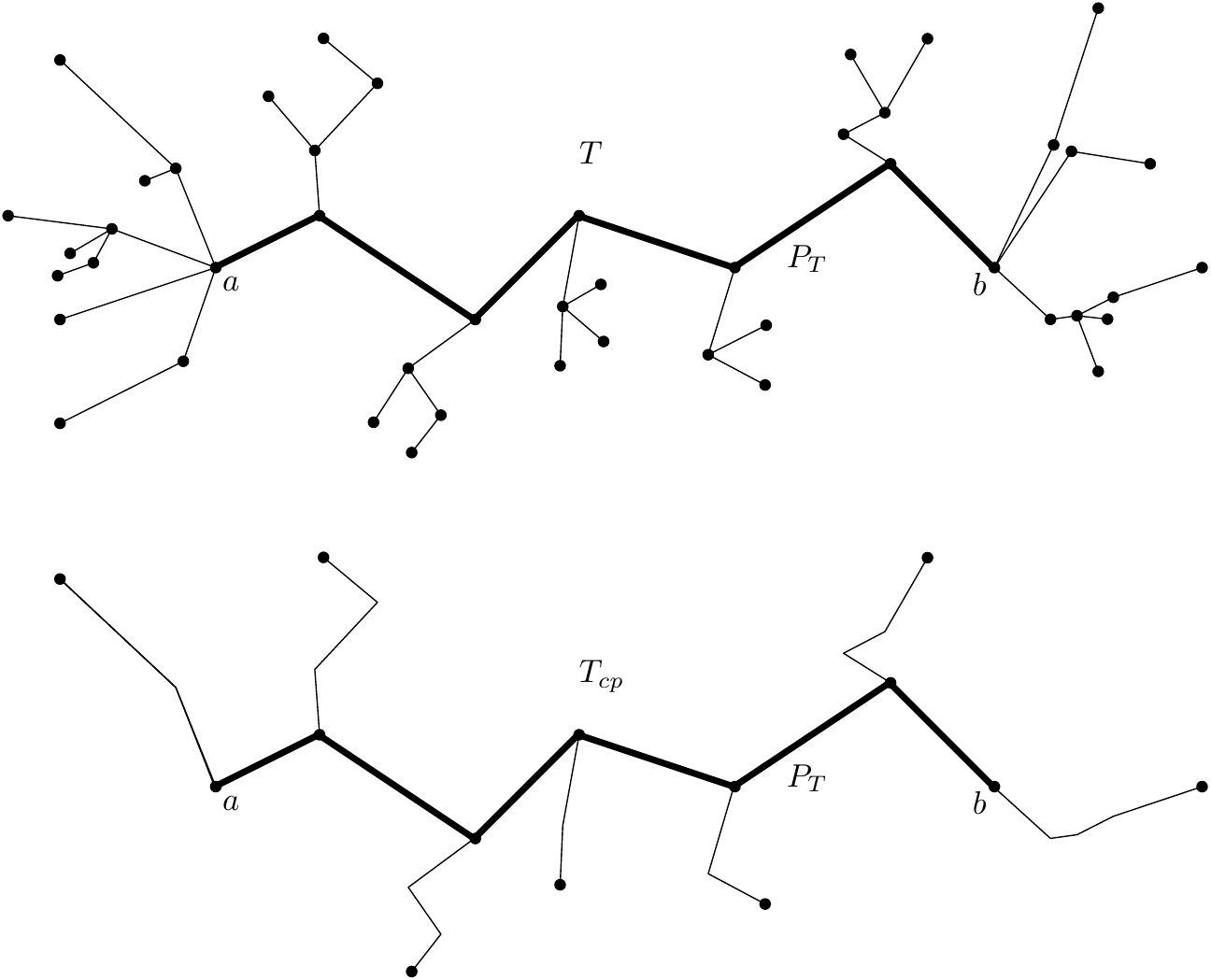}
\caption{Illustrating the conversion of the tree $T$ to the caterpillar
$T_{cp}$, where subtrees \textit{dangling} from $P_T$ (the path from
$a$ to $b$) are compressed to a single edge.}
\label{fig:Caterpillar}
\end{figure}

Recall that for a path, there are only four relevant distances to
compute to determine the diameter; the same holds for a tree with a
unique longest path. These distances can trivially be computed in $O(n)$
time. Now consider the case when one of the endpoints of the shortcut
is fixed at a vertex $v$ and the second endpoint is moving along
$P_{T}$ in $T_{cp}$. As for the path case, the four
functions describing the distances are monotonically increasing or
decreasing, hence, a simple binary search along $P_{T}$ for the second
endpoint can be used to determine the optimal placement of the shortcut.
As a result, the optimal shortcut, given one fixed endpoint $v$ of the
shortcut, can be computed in $O(n \log n)$ time. We get:

\begin{theorem}
Given a tree $T$ on $n$ vertices in a metric space, we can compute a
shortcut that minimizes the diameter of the augmented graph in
$O(n^2 \log n)$ time.
\end{theorem}

Recall that Lemma~\ref{lem:vertices_on_path} states that there exists
an optimal shortcut with both its endpoints on $P_{T}$. However, our
algorithm only requires that one of the endpoints is on $P_{T}$. The
obvious question is if one can modify the algorithm so that it takes
full advantages of the lemma.

\section*{Acknowledgments}
Part of this work was done at the \emph{17th Korean Workshop on
Computational Geometry}, held on Hiddensee Island in Germany,
June 22--27, 2014. We thank the other workshop participants for their
helpful comments.
We also thank Carsten Grimm for his comments on the proof of
Lemma~\ref{lem:vertices_on_path}.

\bibliographystyle{plain}
\bibliography{PathAndTree}

\end{document}